\documentclass[twocolumn,10pt,showpacs,notitlepage,floatfix,superscriptaddress,amsmath,amssymb,
nofootinbib]{revtex4-2}
\usepackage{algpseudocode}
\usepackage{algorithm}
\usepackage[normalem]{ulem}
\usepackage{kpfonts}
\usepackage{graphicx}
\usepackage{amssymb}
\usepackage{amsmath}
\usepackage{amsthm}
\usepackage{amsfonts}
\usepackage{mathtools}
\usepackage{xcolor}
\usepackage{hyperref}
\usepackage{tikz}

\usetikzlibrary{cd}

\newtheorem*{lem*}{Lemma}
\newtheorem*{thm*}{Theorem}
\newtheorem{thm}{Theorem}[section]

\newtheorem{prop}[thm]{Proposition}
\newtheorem{ex}[thm]{Example}

\newtheorem{dfn}[thm]{Definition}

\newcommand{\ket}[1]{|#1\rangle}

\newcommand{\eq}[1]{\hyperref[eq:#1]{Eq.~(\ref*{eq:#1})}}
\renewcommand{\sec}[1]{\hyperref[sec:#1]{Section~\ref*{sec:#1}}}
\newcommand{\secsm}[1]{\hyperref[sec:#1]{Sec.~\ref*{sec:#1}}}
\newcommand{\app}[1]{\hyperref[app:#1]{Appendix~\ref*{app:#1}}}
\newcommand{\theo}[1]{\hyperref[thm:#1]{Theorem~\ref*{thm:#1}}}
\newcommand{\algo}[1]{\hyperref[alg:#1]{Algorithm~\ref*{alg:#1}}}
\newcommand{\lemm}[1]{\hyperref[lem:#1]{Lemma~\ref*{lem:#1}}}
\newcommand{\defn}[1]{\hyperref[defn:#1]{Definition~\ref*{defn:#1}}}
\newcommand{\corr}[1]{\hyperref[cor:#1]{Corollary~\ref*{cor:#1}}}
\newcommand{\condition}[1]{\hyperref[cond:#1]{Condition~\ref*{cond:#1}}}
\newcommand{\fig}[1]{\hyperref[fig:#1]{Fig.~\ref*{fig:#1}}}
\newcommand{\tab}[1]{\hyperref[tab:#1]{Table~\ref*{tab:#1}}}
\newcommand{\tabsm}[1]{\hyperref[tab:#1]{Tab.~\ref*{tab:#1}}}
\newcommand{\propos}[1]{\hyperref[prop:#1]{Proposition~\ref*{prop:#1}}}
\newcommand{\propsm}[1]{\hyperref[prop:#1]{Prop.~\ref*{prop:#1}}}
\newcommand{\rema}[1]{\hyperref[rem:#1]{Remark~\ref*{rem:#1}}}

\newcommand{\hpblsp}{\mathbb{H}^2}
\newcommand{\euclidsp}{\mathbb{R}^2}

\newcommand{\manifold}{\mathcal{M}}
\newcommand{\unicover}{\mathcal{U}}
\newcommand{\im}{{\rm{im}}}

\newcommand{\rs}{{\rm{rs}}}
\newcommand{\transpose}{\mathsf{T}}
\newcommand{\ztwo}{\mathbb{F}_2}

\newcommand{\group}{G}
\newcommand{\nbhood}{\mathcal{A}}
\newcommand{\ancilla}{\mathcal{A}}
\newcommand{\data}{\mathcal{D}}
\newcommand{\vertexset}{\mathcal{V}}
\newcommand{\edgeset}{\mathcal{E}}
\newcommand{\faceset}{\mathcal{F}}

\newcommand{\nocontentsline}[3]{}
\newcommand{\tocless}[2]{\bgroup\let\addcontentsline=\nocontentsline#1{#2}\egroup}

\DeclareMathAlphabet\mathbfcal{OMS}{cmsy}{b}{n} 

\definecolor{darkgreen}{rgb}{0,0.5,0}
\definecolor{red}{rgb}{1,0,0}

\preprint{APS/123-QED}

\begin{document}

\title{Homomorphic Logical Measurements}

 \author{Shilin Huang}
\email{shilin.huang@duke.edu}
 \affiliation{Duke Quantum Center, Duke University, Durham, NC 27701, USA}
 \affiliation{Department of Electrical and Computer Engineering, Duke University, Durham, NC 27708, USA}

 \author{Tomas Jochym-O'Connor}
 \affiliation{IBM Quantum, IBM Almaden Research Center,  San Jose,  CA 95120,  USA}
 \author{Theodore J. Yoder}
 \affiliation{IBM Quantum, IBM T.J. Watson Research Center, Yorktown Heights, NY 10598, USA}

\begin{abstract}
Shor and Steane ancilla are two well-known methods for fault-tolerant logical measurements, which are successful on small codes and their concatenations.
On large quantum low-density-parity-check (LDPC) codes, however, Shor and Steane measurements have impractical time and space overhead respectively.
In this work, we widen the choice of ancilla codes by unifying Shor and Steane measurements into a single framework, called \textit{homomorphic measurements}.
For any Calderbank-Shor-Steane (CSS) code with the appropriate ancilla code,
one can avoid
repetitive measurements or complicated ancilla state preparation procedures such as distillation, which overcomes the difficulties of both Shor and Steane methods.
As an example, we utilize the theory of covering spaces to construct homomorphic measurement protocols for arbitrary $X$- or $Z$-type logical Pauli operators on surface codes in general, including the toric code and hyperbolic surface codes.
Conventional surface code decoders, such as minimum-weight perfect matching, can be directly applied to our constructions.

\end{abstract}

\maketitle

\section{Introduction}

Quantum error correction \cite{Shor:1995b, kitaev2003fault, terhal2015quantum} is a necessary component for large-scale digital quantum computation \cite{Shor:1995,lloyd1996universal}.
Fault-tolerant quantum computation (FTQC) via surface codes \cite{kitaev2003fault,
bravyi1998quantum,Dennis:2002,raussendorf2007fault} has been intensively optimized \cite{Horsman_2012,litinski2019game,Fowler:2012} due to its several experimentally friendly features, including nearest-neighbor 2D connectivity \cite{Fowler:2012}, availability of efficient decoding algorithms \cite{Dennis:2002,edmonds2009paths, Wang:2009, fowler2012towards, delfosse2021almost,huang2020fault,higgott2022fragile} and high fault-tolerance threshold values \cite{wang2011surface,higgott2022fragile}. However, the resource requirement for surface code FTQC is still daunting as the encoding rate vanishes in the asymptotic limit~\cite{BravyiPoulinTerhal}.

In principle, FTQC via quantum low-density-parity-check (LDPC) codes of non-vanishing encoding rate can achieve constant space overhead and therefore be more resource efficient \cite{Gottesman:2014, Fawzi2018constant}. 
While long-range interactions are unavoidable for these codes \cite{baspin2022quantifying, baspin2022connectivity,
tremblay2022constant, delfosse2021bounds}, 
as demonstrated by recent experiments \cite{stephenson2020high,zhang2022scalable, bluvstein2022quantum}, connectivity is not a fundamental hardware constraint.
Recent breakthroughs on quantum LDPC codes are encouraging:
Theoretically, a series of results \cite{hastings2021fiber,breuckmann2021balanced,panteleev2021quantum} led to the exciting discovery of ``good'' quantum LDPC codes, in the sense that non-vanishing encoding rate and linear distance are simultaneously achieved \cite{panteleev2002asymptotically,leverrier2022quantum,dinur2022good}; Practically, 
numerical studies have shown fault-tolerance threshold values of quantum LDPC codes above $10^{-3}$ \cite{tremblay2022constant}, while the decoders and syndrome extraction circuits have not been fully optimized yet.

How fault-tolerant logical operations are performed on quantum LDPC codes is still far from well understood. 
As a proof-of-concept step, techniques from surface code FTQC, such as code deformation \cite{Bombin_2009} and lattice surgery \cite{Horsman_2012}, are firstly migrated to specific instances such as hyperbolic surface codes \cite{Breuckmann2017hyperbolic, Lavasani2018} and hypergraph product codes \cite{tillich2014quantum, krishna2021hypergraph,Krishna2020wormhole},
then further generalized to any Calderbank-Shor-Steane (CSS, \cite{calderbank1997orthgonal}) quantum LDPC codes \cite{cohen2022lowoverhead}.
Importantly, \textit{individual addressing} of logical qubits is enabled by the capability of performing arbitrary logical Pauli measurements \cite{cohen2022lowoverhead}.
However, the code deformation processes oftentimes do not preserve the decoding properties of given quantum LDPC codes at the quantum memory stage, such as the ability of applying minimum-weight-perfect-matching (MWPM) decoders \cite{Dennis:2002,edmonds2009paths,Fowler:2012} or single-shot error correction \cite{bombin2015single, Campbell2019,Fawzi2018constant}.   

If connectivity is not a limitation for quantum LDPC codes, logical operations based on transversal two-qubit gates can also be considered as a candidate approach. 
The first example was developed by Shor~\cite{shor1996fault,divincenzo1996fault}, in which any logical Pauli operator is measured by an ancilla cat state of the same weight.
While Shor measurement supports individual addressing of logical qubits, the measurement circuit has to be repeated to boost the accuracy. As the code distance increases, the number of repetitions grows exponentially  \cite{Lavasani2018}, which is unfavorable for large quantum LDPC codes.
To avoid repetitive readouts, 
Steane~\cite{steane1997active} and Knill~\cite{Knill2005nature} utilized ancilla blocks encoded by the same code as data to perform single-shot logical measurements. 
However, to address a particular logical qubit, the ancilla block needs to be prepared in specific logical state, whose preparation requires protocols of large space overhead such as distillation \cite{Zheng_2018,Zheng_2020} and code concatenation \cite{Knill2005nature, Gottesman:2014}. 

In this paper, inspired by Refs.~\cite{Huang2021PRL,huang2021constructions}, 
we unify Shor and Steane measurements for CSS codes into a single framework, called \textit{homomorphic measurements}. In our framework, the data code block can interact with an ancilla block encoded in another CSS code. By choosing an appropriate ancilla code that encodes only one logical qubit and has the same distance as the data code, we are able to perform single-shot measurements of arbitrary single- and multi-qubit CSS logical Paulis without the need of state distillation, which circumvent the difficulties of both Shor and Steane measurements.
To show that constructing such ancilla codes is possible,
we provide a procedure for constructing these ancilla codes for surface codes on two-dimensional closed manifolds including Kitaev's toric code~\cite{kitaev2003fault} and hyperbolic surface codes~\cite{breuckmann2016constructions,Breuckmann2017hyperbolic,higgott2021subsystem}. 
Our construction utilizes the concept of \textit{covering spaces} in topology \cite{Hatcher}. In contrast to logical measurement methods based on lattice surgery \cite{Breuckmann2017hyperbolic,cohen2022lowoverhead}, our approach supports conventional decoding algorithms including MWPM~\cite{Dennis:2002,edmonds2009paths, Fowler:2012} and union-find~\cite{delfosse2021almost}
in a straightforward way. 

The paper is organized as follows. \sec{background} provides essential background and present our notation.
\sec{Shor_and_Steane} reviews Shor and Steane measurements and discusses their advantages and drawbacks.  \sec{homo_meas} presents our homomorphic measurement framework.
\sec{homsurface} provides the procedure for constructing homomorphic measurement circuits for surface codes,
and compares it to the existing approaches. 
\sec{conclusion} concludes and discusses future research directions.

\section{Notations and Background}\label{sec:background}

\subsection{Binary vector spaces}

Given an $r\times n$ binary matrix $H: \ztwo^n \rightarrow \ztwo^r$, the transpose, kernel (null space), image (column space) and row space of $H$ are denoted by $H^\transpose$, $\ker(H)$, $\im(H)$ and $\rs(H)$, respectively.
The binary vector space spanned by some finite set $S$ is denoted by $\ztwo[S]$.
Note that each vector $v = \sum_{e\in S} v_e e \in \ztwo[S]$ can be viewed as a subset of $S$, where each $e \in S$ is in the subset if and only if $v_e = 1$.
The transpose of a linear map $H: \ztwo[A] \rightarrow \ztwo[B]$ is defined to be the map $H^\transpose: \ztwo[B] \rightarrow \ztwo[A]$ such that the matrix $H^\transpose$ is the transpose of $H$ under the bases $A$ and $B$.

\subsection{CSS stabilizer codes}

An $[[n,k,d]]$ CSS code $(H_X,H_Z)$ is defined by two parity-check matrices $H_X: \ztwo^n \rightarrow \ztwo^{r_X}$ and $H_Z: \ztwo^n \rightarrow \ztwo^{r_Z}$ such that $H_X H_Z^\transpose = 0$. 
The quantum code has an $X$-type Pauli stabilizer group isomorphic to the space $\rs(H_X)$, and a $Z$-type stabilizer group isomorphic to $\rs(H_Z)$. The $X$-type and $Z$-type logical operators are represented by elements of $\ker(H_Z)$ and $\ker(H_X)$, respectively.
The number of encoded logical qubits 
\begin{eqnarray*}
k &=& \dim\left(\ker(H_X)/\rs(H_Z)\right)\\
&=& \dim\left(\ker(H_Z)/\rs(H_X)\right).
\end{eqnarray*}
The $X$-distance ($Z$-distance), denoted by $d_X$ ($d_Z$), is the minimum weight of $X$-type ($Z$-type) logical errors. We have 
\begin{eqnarray*}
d_X &:=& \min\left\{|c|: c \in \ker(H_Z)\setminus \rs(H_X)\right\},\\
d_Z &:=& \min\left\{|c|: c \in \ker(H_X)\setminus \rs(H_Z)\right\}.
\end{eqnarray*}
The code distance $d = \min\{d_X, d_Z\}$.

There is a well-established intepretation of CSS codes as \textit{chain complexes}~\cite{freedman2001projective,kitaev2003fault}. 
In this work, a chain complex is a sequence of binary vector spaces $\{C_i\}$ connected by some linear maps $\left\{\partial_i\right\}$:
\begin{eqnarray*}
\dots \rightarrow C_{i+1} \xrightarrow{\partial_{i+1}} C_i \xrightarrow{\partial_i} C_{i-1} \rightarrow \dots
\end{eqnarray*}
satisfying the condition $\partial_{i-1} \partial_i = 0$.
The maps $\left\{\partial_i\right\}$ are referred to as the \textit{boundary maps}.
Any CSS code $(H_X, H_Z)$ can be viewed as a length-$3$ chain complex
\begin{eqnarray*}
C_2 \xrightarrow{\partial_2}  C_1  \xrightarrow{\partial_1}  C_0 
\end{eqnarray*}
with $\partial_2 = H_Z^\transpose$ and $\partial_1 = H_X$.
The condition $\partial_1 \partial_2 = H_XH_Z^\transpose = 0$ is satisfied by definition.

\section{Fault-tolerant logical measurement via Shor and Steane methods}
\label{sec:Shor_and_Steane}

In this section, we provide a brief overview of Shor and Steane measurements. 
While both schemes can be generalized for logical Pauli measurements of non-CSS stabilizer codes \cite{Zheng_2018,Zheng_2020,yoder2018practical}, our discussion focuses on measuring $Z$-type logical Pauli operators of CSS codes.
We further discuss their limitations on large block codes.  

\subsection{Shor measurement}

Shor measurement uses an ancilla cat state to measure the logical Pauli operator directly, without utilizing the structure of the data code block.
For example, to measure a $Z$-type Pauli operator $\overline{P}$ of weight $w$, one can prepare a $w$-qubit ancilla cat state $\frac{1}{\sqrt{2}} \left(\ket{+^{\otimes w}} + \ket{-^{\otimes w}}\right)$, apply a transversal CNOT gate between the support of $\overline{P}$ (as control) and the cat state (as target), and finally measure all ancilla qubits in $Z$-basis. 
The outcome of measurement $\overline{P}$ is obtained via taking the parity of the individual measurement outcome.
Assuming the cat states are prepared fault-tolerantly in the sense that no correlated $Z$-errors exist,
the transversal CNOT gates do not create correlated errors on the data block.
The construction can be easily generalized to measure arbitrary Pauli operators.

The cat state measurement is not yet fault-tolerant as 
any individual measurement error can flip the logical measurement outcome. 
One might consider repeating the cat state measurement to gain enough confidence of the logical outcome by a majority vote.
However, life is not that simple: any single Pauli error on the data block that anti-commutes with $\overline{P}$ will flip all of the repeated logical measurement outcomes.
To resolve this issue, we need to apply fault-tolerant error correction between two consecutive cat-state measurements.  
Another idea to achieve fault tolerance is to choose different representatives of $\overline{P}$ in different measurement rounds \cite{delfosse2020short}.

\subsection{Steane measurement}

Steane measurement utilizes the fact that transversal CNOT gates are available for CSS codes. Instead of cat states, one can use the appropriate codeword of the CSS code as the ancilla state.

We start with the case that the CSS code $(H_X,H_Z)$ only encodes one logical qubit. To measure the logical $Z$ operator of the data block, one can prepare an ancilla code block in logical state $\overline{\ket{0}}$, apply a transversal CNOT gate between data (control) and ancilla (target), and measure all qubits in $Z$-basis.
To read out the logical outcome, one applies error correction of the classical code defined by $H_Z$, then takes the parity of the individual outcomes inside the support of any representative of the logical $Z$ operator.
Unlike Shor measurement, Steane measurement does not require any repetitions. This is because the ancilla block can correct up to $\lfloor \frac{d-1}{2} \rfloor$ measurement errors. The $X$-errors on data qubits before the transversal CNOT are also not a concern, as they are equivalent to measurement errors.

More generally, if the data block encodes $k>1$ logical qubits, to measure a general $Z$-type logical operator $\overline{P}$, the ancilla state should be stabilized by $\overline{P}$ and any other $X$-type logical operators that commute with $\overline{P}$. As an example, suppose $k = 3$ and we want to measure the logical $\overline{Z_1Z_2Z_3}$ of the data block. The ancilla code state should be stabilized by $\overline{Z_1Z_2Z_3}$, $\overline{X_1X_2}$ and $\overline{X_2X_3}$, which is the logical cat state $$\frac{\overline{\ket{+_1+_2+_3}}+\overline{\ket{-_1-_2-_3}}}{\sqrt{2}}.$$

\subsection{Limitations of Shor and Steane measurements}

Shor measurement, while straightforward, can have an impractical time overhead for large quantum codes. 
The reason is that to obtain meaningful logical outcome,
even without counting the errors in ancilla state preparation,
the individual measurement error rate $p$ needs to satisfy $pw \ll 1$.
However, when $p$ is a constant for a quantum code family, as the distance $d$ increases, we will eventually have $pw \ge pd \gg 1$. As pointed out in Ref. \cite{Lavasani2018}, the probability of measuring the logical outcome correctly is
$\frac{1}{2} + \frac{1}{2} (1-2p)^w$,
which is exponentially close to $1/2$ as $d$ increases. 
As a result, a fault-tolerant logical readout will require an exponential number of repetitions. 

While Steane measurement is single-shot, the complication of Steane ancilla preparation is a concern. 
A common strategy is to distill good ancilla states without correlated errors from noisy ones with non-fault-tolerant preparations \cite{reichardt2004improved, Zheng_2018,Zheng_2020}.
However, if an ancilla block is prepared by a non-fault-tolerant unitary encoding circuit, there can be typically $O(np)$ number of correlated errors of arbitrary weights, where $p$ is the physical error rate. 
As the code size $n$ increases,
the rejection rate of postselection will be exponentially close to $1$ in order to distill usable ancilla blocks.
To overcome this issue, one has to apply code concatenation~\cite{knill2005scalable,Gottesman:2014} to reduce the physical error rate $p$ to at least $p/n$ before state distillation, which result in non-negligible space overhead.

\section{Homomorphic logical measurements}\label{sec:homo_meas}

For Steane measurement, there is a special case that postselection can be fully avoided:
For quantum LDPC codes of $k=1$, 
the Steane ancilla state $\overline{\ket{0}}$ can
be prepared via \textit{Shor error correction}~\cite{shor1996fault}.
The idea is to first initialize all the ancilla qubits in $\ket{0}$, then repeat measuring the stabilizer elements for some number of rounds and apply error correction.
Each stabilizer element can be measured by a cat-state gadget or other simplified constructions~\cite{divincenzo2007effective,Chao2018,chamberland2018flag,Yoder:2017b,Li:2017,Chamberland2020topo,huang2019fault,tansuwannont2020flag}. 
Note that the logical $Z$ operator and $Z$-type stabilizer generators are already in the stabilizer group before the repetitive stabilizer measurements. 
For this reason, measuring the high-weight logical $Z$ operator is not required. 
For many quantum LDPC code families, with appropriate choices of decoders and number of repetitive measurement rounds, logical state preparation via Shor error correction has a fault tolerance threshold and is robust against \textit{typical errors} \cite{Dennis:2002,Fawzi2018constant}.

For codes encoding $k>1$ logical qubits,
it is straightforward to generalize the above protocol to prepare $\overline{\ket{0^{\otimes k}}}$ and $\overline{\ket{+^{\otimes k}}}$. However, for other logical CSS ancilla states, both $Z$- and $X$-type high-weight stabilizer elements exist. As a result, one cannot avoid high-weight stabilizer measurements 
by initializing all ancilla qubits in $\ket{0}$ or $\ket{+}$ before Shor error correction. 
A natural thought on circumventing this difficulty is to choose another $[[m,1,d']]$ quantum LDPC code as ancilla. In fact, Shor measurement is already an example, where the ancilla code is the quantum repetition code with weight-$2$ $X$-type stabilizers. It is just unfortunate that quantum repetition code has distance $d' = 1$, which results in the exponential time overhead  of Shor measurement. To avoid repetitive measurements, it is desirable to have $d' = d$.

From the above observations, we are ready to ask the following question: \textit{can we construct an $[[m,1,d]]$ ancilla quantum LDPC code for single-shot fault-tolerant $Z$-type logical measurement?} To answer this, we need to understand what kinds of ancilla gadgets, including the choice of ancilla codes and how the data interacts with ancilla, are allowed for a given data code.
As we will see later, these ancilla gadgets can be described as a homomorphism between two chain complexes representing the ancilla and data CSS codes. 
For this reason, we refer these ancilla gadgets to as \textit{homomorphic gadgets}. 

\subsection{Homomorphic gadgets: definition}

Suppose we have a data block encoded in an $[[n,k,d]]$ CSS code $(H_X,H_Z)$ and we apply some CNOT gates to entangle the data block with an $[[m,k',d']]$ ancilla CSS code block $(H_X',H_Z')$.
The CNOT gates are encoded as a matrix $\Gamma: \ztwo^{m} \rightarrow \ztwo^n$ referred to as the \textit{gate matrix}, where  $\Gamma_{ij} = 1$ if and only if there is a CNOT gate between $i$-th data qubit as control and $j$-th ancilla qubit as target.
To simplify our discussion, the homomorphic gadgets we developed here are for $Z$-type logical measurements. For $X$-type logical measurements, the ancilla qubits serve as controls while the data qubits are targets.

The data-ancilla interaction, represented by $\Gamma$, should be logical in the sense that it preserves the total stabilizer group of two blocks. 
Before the interaction, the total 
$Z$- and $X$-stabilizer groups are isomorphic to the vector spaces
$$T_Z := \rs\left(\begin{array}{cc}H_Z &
0\\
0 & H_Z'\end{array}\right) \subseteq \ztwo^{n+m}$$ 
and 
$$T_X := \rs\left(\begin{array}{cc}H_X &
0\\
0 & H_X'\end{array}\right) \subseteq \ztwo^{n+m}$$ 
respectively. 
After the interaction, they become
$$T_Z' := \rs\left(\begin{array}{cc}H_Z & 
0\\
H_Z' \Gamma^\transpose &
H_Z' \end{array}\right)$$ 
and
$$T_X' := \rs\left(\begin{array}{cc}H_X & 
H_X\Gamma\\
0& H_X'
\end{array}\right)$$ 
respectively. 
To guarantee $T_Z' = T_Z$ and $T_X' = T_X$, we must have
\begin{eqnarray}
\rs(H_Z' \Gamma^\transpose) \subseteq \rs(H_Z)\label{eqn:homgadget1}
\end{eqnarray}
and
\begin{eqnarray}
\rs(H_X \Gamma) \subseteq \rs(H_X'). \label{eqn:homgadget2}
\end{eqnarray}

From the above discussion, we give our formal definition of homomorphic gadgets:
\begin{dfn}[Homomorphic gadget]
An $[[m,k',d']]$ homomorphic gadget $(H_X',H_Z',\Gamma)$ for an $[[n,k,d]]$ CSS code $(H_X,H_Z)$ has 
the following data:
\begin{itemize}
    \item[(i)] An ancilla $[[m,k',d']]$ CSS code $(H_X',H_Z')$,
    \item[(ii)] The gate matrix $\Gamma: \ztwo^m \rightarrow \ztwo^n$,
\end{itemize}
where $H_X'$, $H_Z'$ and $\Gamma$ satisfy 
$\rs(H_Z \Gamma^\transpose) \subseteq \rs(H_Z)$ and $\rs(H_X \Gamma) \subseteq \rs(H_X')$.
\end{dfn}

\subsection{Logical measurements}

To measure $Z$-type logical operators of the $[[n,k,d]]$ data code block $(H_X,H_Z)$ using an $[[m,k',d']]$ homomorphic gadget $(H_X',H_Z',\Gamma)$, 
one first prepares the ancilla code $(H_X',H_Z')$ in the logical state $\overline{\ket{0^{\otimes k'}}}$, applies the data-ancilla interaction encoded by $\Gamma$, and finally measures all ancilla qubits in the $Z$-basis. 

Before the data-ancilla interaction, the $Z$-stabilizer group of the ancilla block is isomorphic to the vector space $\ker H_X'$ as any $Z$-type logical operators of the ancilla stabilize $\overline{\ket{0^{\otimes k'}}}$. 
After the interaction,
each $Z$-stabilizer element $v \in \ker H_X'$ is transformed as
a $Z$-type operator 
$\Gamma v \oplus v \in \ztwo^{n+m}$
supported on both data and ancilla blocks.
If $v \in \rs(H_Z')$ is in the $Z$-stabilizer group of the ancilla code $(H_X',H_Z')$, by our definition of homomorphic gadgets,
$v$ will still be a $Z$-stabilizer of the ancilla block after the interaction.
After the transversal $Z$-measurements, the outcome of $v$ can be used as a parity check to correct $X$ errors.
If $v \notin \rs(H_Z')$, however, 
the $Z$-type operator $\Gamma v$ of the data block is being measured, whose outcome is the measurement outcome of $v$.
The group of $Z$-type logical operators measured by the gadget is therefore isomorphic to $\Gamma(\ker H_X')$. 
Typically, we choose $k'=1$ to  measure just one $Z$-type logical operator $\Gamma v$, where $v \in \ker H_X'$ is the unique non-trivial element. 
If necessary, we can choose $k'>1$ to perform parallel logical measurements.

\subsection{Homomorphic gadgets as chain homomorphisms}

We now address the naming of homomorphic gadgets.
As discussed in \sec{background}, we can describe the data and ancilla CSS code blocks as two chain complexes
\begin{eqnarray*}
C_2 \xrightarrow{\partial_2}  C_1  \xrightarrow{\partial_1}  C_0 
\end{eqnarray*}
and 
\begin{eqnarray*}
C_2' \xrightarrow{\partial_2'}  C_1'  \xrightarrow{\partial_1'}  C_0'
\end{eqnarray*}
respectively. The data-ancilla interaction is a linear map $\gamma_1: C_1' \rightarrow C_1$. 
In the previous notation, the gadget is the tuple $(\partial_1', \partial_2'^\transpose, \gamma_1)$.
The conditions (\ref{eqn:homgadget1}) and (\ref{eqn:homgadget2}) are rewritten as
\begin{eqnarray}
\im(\gamma_1 \partial_2') \subseteq \im(\partial_2) \label{eqn:homgadget1chain}
\end{eqnarray}
and 
\begin{eqnarray}
\rs(\partial_1\gamma_1) \subseteq \rs (\partial_1') \label{eqn:homgadget2chain}
\end{eqnarray}
respectively.
Eq.(\ref{eqn:homgadget1chain}) implies that there exists a linear map $\gamma_2: C_2' \rightarrow C_2$ such that $\partial_2\gamma_2 = \gamma_1\partial_2'$. Similarly, Eq.(\ref{eqn:homgadget2chain}) implies that there exists a linear map $\gamma_0: C_0' \rightarrow C_0$ such that
 $\gamma_0 \partial_1' = \partial_1 \gamma_1$. 
Therefore we obtain a commutative diagram
\begin{eqnarray*}
\begin{tikzcd}
 C_2' \arrow[d,"\gamma_2"]\arrow[r, "\partial_2'"] & C_1' \arrow[d,"\gamma_1"]\arrow[r, "\partial_1'"] & C_0' \arrow[d,"\gamma_0"] \\
 C_2 \arrow[r,"\partial_2"]& C_1 \arrow[r,"\partial_1"] & C_0 
\end{tikzcd}
\end{eqnarray*}
The series of mappings $\{\gamma_i: C_i' \rightarrow C_i\}$ is called a \textit{chain homomorphism} from $\{C_i'\}$ to $\{C_i\}$ \cite{Hatcher}.

One might wonder the usefulness of the chain homomorphism formulation, as the maps $\gamma_0$ and $\gamma_2$ are not used for specifying a homomorphic gadget. 
However, as many quantum LDPC codes are constructed from homology of manifolds \cite{kitaev2003fault,breuckmann2016constructions} or more general homological algebra \cite{tillich2014quantum,breuckmann2021balanced,hastings2021fiber,panteleev2002asymptotically,panteleev2021quantum}, such formulation will be convenient for finding explicit homomorphic gadget constructions.

\section{Homomorphic measurements of surface codes}\label{sec:homsurface}

We have not yet shown how to construct a homomorphic gadget with $[[m,1,d]]$ LDPC ancilla to measure a desired $Z$-type logical operator of an $[[n,k,d]]$ CSS LDPC code. 
Unlike Shor and Steane measurements, it is difficult to come up with a universal construction. 
For a proof-of-principle demonstration,
we show that for surface codes defined on a 2D closed manifold,
including toric codes \cite{kitaev2003fault}, hyperbolic surface codes, \cite{breuckmann2016constructions} and variants \cite{Breuckmann2017hyperbolic,higgott2021subsystem},
one can construct $[[m,1,d]]$ homomorphic gadgets for measuring arbitrary $Z$-type logical operators, where the ancilla codes are also surface codes.

\subsection{Surface codes}

In this section, a surface code is defined a by cellulation of a manifold $\manifold$ of dimension no more than $2$, denoted by $\manifold = (\vertexset, \edgeset, \faceset)$, where $\vertexset$, $\edgeset$ and $\faceset$ are the sets of $0$-cells (vertices), $1$-cells (edges) and $2$-cells (faces) of $\manifold$, respectively. 
A chain complex
\begin{eqnarray*}
\ztwo[\faceset] \xrightarrow{\partial_2} \ztwo[\edgeset] \xrightarrow{\partial_1} \ztwo[\vertexset]
\end{eqnarray*}
can be constructed,  where for each $f \in \faceset$, $\partial_2 f$ is the sum of edges that borders $f$, 
and for each $e \in \edgeset$, 
$\partial_1 e$ is the sum of two endpoints of $e$. As $\partial_2 f$ is a closed loop, we have $\partial_1 \partial_2 f = 0$ for all $f \in \faceset$. Therefore $\partial_1\partial_2 = 0$ and a CSS code $(\partial_1, \partial_2^\transpose)$ can be defined, where $\vertexset$, $\edgeset$ and $\faceset$ are served as the sets of $X$-checks, qubits and $Z$-checks, respectively.
Roughly speaking, a $Z$-type ($X$-type) logical operator corresponds to a non-contractible loop (cut) of the cellulation of the manifold $\manifold$. 

Note that the manifold $\manifold$ is allowed to have boundary. 
To simplify our discussion, for the majority of this section, we only allow the existence of  \textit{smooth boundaries} so that $Z$-type logical operators are always loops instead of a path connecting two different \textit{rough boundaries}. 
As an example, a surface code defined on a square does not encode any logical qubit, which differs from the planar surface code \cite{bravyi1998quantum} encoding one logical qubit. On a planar surface code, two opposite sides of the square are rough boundaries, while the other two sides are smooth. 
We will discuss surface codes with rough boundaries in \sec{roughboundary}.

We also note that when $\dim \manifold = 1$, there will be no $Z$-checks as $\faceset = \emptyset$. In particular, when $\manifold$ is a circle, the obtained code is a quantum repetition code with two-body $X$-type stabilizer elements.

A general idea for constructing homomorphic gadgets for a surface code $\data = (\vertexset, \edgeset, \faceset)$ is to find another surface code $\ancilla = (\vertexset', \edgeset', \faceset')$ 
and a continuous mapping $\gamma: \ancilla \rightarrow \data$ that preserves the cellulation structure. Then, $\gamma$ induces a chain homomorphism
\begin{eqnarray*}
\begin{tikzcd}
 \ztwo[\faceset'] \arrow[d,"\gamma_2"]\arrow[r, "\partial_2'"] & \ztwo[\edgeset'] \arrow[d,"\gamma_1"]\arrow[r, "\partial_1'"] & \ztwo[\vertexset'] \arrow[d,"\gamma_0"] \\
 \ztwo[\faceset] \arrow[r,"\partial_2"]& \ztwo[\edgeset] \arrow[r,"\partial_1"] & \ztwo[\vertexset]
\end{tikzcd}
\end{eqnarray*}
between the chain complexes associated with $\nbhood$ and $\data$. Here $\partial_i'$ are the boundary maps of $\nbhood$, and $\gamma_i$ are linear maps induced by $\gamma$ such that 
$\gamma_0(v') = \gamma(v')$, $\gamma_1(e') = \gamma(e')$ and $\gamma_2(f') = \gamma(f')$ for any vertices $v' \in \vertexset'$, edges $e' \in \edgeset'$ and faces $f' \in \faceset'$, respectively.
The resulting homomorphic gadget $(\partial_1', \partial_2'^\transpose, \gamma_1)$ has an ancilla CSS code $(\partial_1', \partial_2'^\transpose)$ and a gate matrix $\Gamma=\gamma_1$.

\subsection{Homomorphic gadgets from subspaces}
\begin{figure*}
    \centering
    \includegraphics[width=0.9\linewidth]{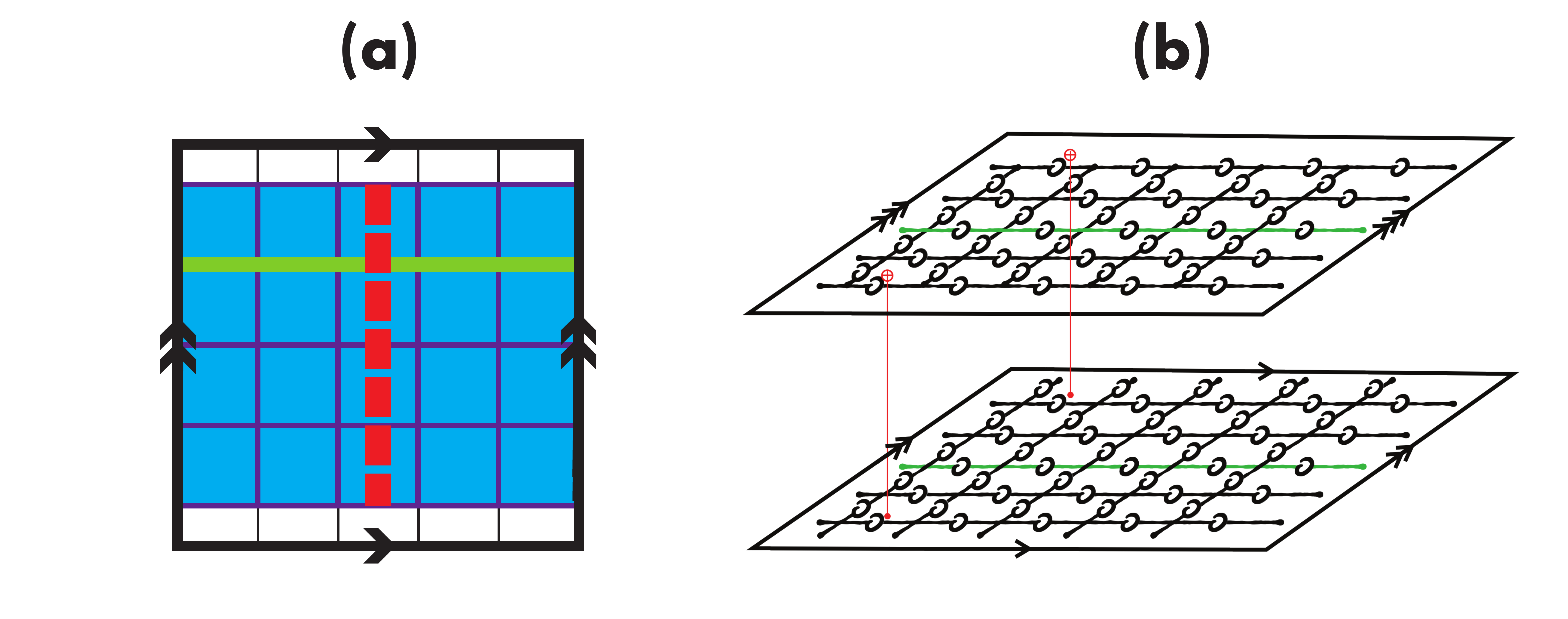}
    \caption{
    (a) Homomorphic gadget for measuring $\overline{Z}_1$ of the standard $[[2d^2,2,d]]$ toric code. The green line is the loop $\ell$ corresponding to $\overline{Z}_1$. The ancilla code $\ancilla$ is a ribbon encoding only one logical qubit. The faces ($Z$-stabilizer elements) of $\ancilla$ are colored in blue.  The logical $X$ operator of $\ancilla$ is formed by the edges on the cut in red. (b) Full implementation of the homomorphic measurement. The data and ancilla qubits are on the bottom and top layers respectively. The logical operator is measured by performing a CNOT gate between the data qubits (as control) and each corresponding ancilla qubit (as target). Note that there are $d$ data qubits (vertical edges) without corresponding ancilla qubits.
    }
    \label{fig:toric_code_Z1}
\end{figure*}

We now investigate how to construct a homomorphic gadget to measure some $Z$-type logical operator $\overline{P}$ of a surface code  $\data = (\vertexset,\edgeset,\faceset)$ of distance $d_\data$. For simplicity,
 we first assume that $\overline{P}$ correspond to a single loop $\ell$ on the graph $(\vertexset, \edgeset)$. The case that $\overline{P}$ corresponding to multiple disjoint loops will be addressed in \sec{joint_meas}.

A straightforward idea is to pick a 
subcomplex $\nbhood = (\vertexset',\edgeset',\faceset')$ of $\data$ such that $\ancilla \supseteq \ell$, 
then utilize the inclusion map $\gamma: \ancilla \hookrightarrow \data$ to induce the homomorphic gadget. As $\gamma$ is injective, the data-ancilla interaction is \textit{transversal} and immediately fault-tolerant.
An extreme case is that $\nbhood = \ell$ does not contain any faces, i.e., $\faceset' = \emptyset$.
If $\ell$ is not self-intersecting, the homomorphic gadget corresponds to the cat-state gadget for Shor measurement.

To avoid repetitive logical measurements, we need to enlarge the subcomplex $\ancilla$ and increase the code distance $d_\ancilla$ of $\ancilla$ to $d_\data$.
As shown in \fig{toric_code_Z1}, we find 
that when $\data$ is the standard toric code and $\ell$ is the horizontal loop on the torus corresponding to the logical $\overline{Z}_1$ operator, $\ancilla$ can be chosen as the product of a horizontal circle of length $d_\data$ and a vertical segment of length $d_\data-1$.

While we look for $\nbhood$ of distance $d_\nbhood = d_\data$, we also want to ensure that $\nbhood$ encodes only one logical qubit so that the ancilla state can still be easily prepared. 
It is clear that $\nbhood$ cannot be arbitrarily large: the largest subcomplex $\nbhood = \data$ encodes as many logical qubits as $\data$ does, which corresponds to Steane measurement of all $Z$-type logical operators rather than just the single logical operator we wanted.
In fact, there exist cases where it is even impossible to find any subcomplex $\nbhood \supseteq \ell$ of $\data$ that only encodes one logical qubit:
if $\ell$ is a self-intersecting loop which can be further decomposed into two non-contractible loops $\ell_1$ and $\ell_2$, 
even the smallest choice of subcomplex $\nbhood = \ell$ will encode two logical qubits. 

\subsection{Homomorphic gadgets from  covering spaces}\label{sec:homgadget_coveringspace}

To circumvent the difficulties discussed above, it would be nice if one can ``unfold'' the manifold $\data$ in a way that $\ell$ represents the unique $Z$-type logical operator up to stabilizers, while the other $Z$-type logical loops of $\data$ are unfolded. 
One can therefore choose a large enough subcomplex of the newly obtained space without worrying that other $Z$-type logical operators will be measured.
As we will see in this section, such intuition can be realized by the concept of \textit{covering spaces}. 

\subsubsection{Groups acting on spaces}

Before introducing covering spaces, we find it useful to discuss how the manifolds $\manifold$ for defining surface codes are constructed in general, using an abstract formulation known as \textit{groups acting on spaces}. 
Our treatment mainly follows Refs. \cite{breuckmann2017homological, Hatcher}.

Let $\unicover$ be a cellulation of a \textit{simply-connected} two-dimensional manifold, on which all loops are contractible.
Let $G$ be a group of homeomorphisms from $\unicover$ to itself such that for each element $g \in G$, the map $g: \unicover \rightarrow \unicover$ preserves the cellulation structure.
We can define a space $\unicover / \group$ by identifying each point $u \in \unicover$ with all images $g(u)$, where $g$ runs through all elements of $\group$.
The elements of $\unicover / \group$ are therefore the \textit{orbits} $Gu := \{g(u)|g \in G\}$ in $\unicover$, and $\unicover / \group$ is called the \textit{orbit space} of the action of $G$. 
Define the quotient map $p_\group: \unicover \rightarrow \unicover / \group$ by $p_\group(u) = \group u$ for all $u \in \unicover$.
As the elements of $G$ preserves the cellulation structure of $\unicover$, the space $\manifold = \unicover / \group$ inherits a cellulation structure from the map $p_G$: for each $k$-cell $e$ of $\unicover$, $p_G(e)$ is a $k$-cell of $\manifold$.

\begin{ex}[Toric code]
To construct the standard $[[2d^2,2,d]]$ toric code, one can take $\unicover = \euclidsp$ and $\group$ be the set of translations $t_{r,s}: \euclidsp \rightarrow \euclidsp$, where $(r,s) \in \mathbb{Z}^2$ and 
\begin{eqnarray}
t_{r,s}(x,y) := (x+dr, y+ds)
\end{eqnarray}
for each $(x,y) \in \euclidsp$.
It is evident that $\group \cong \mathbb{Z}^2$ and $\manifold = \unicover / \group$ is a torus.
If we cellularize $\unicover$ as a square lattice, where the faces are the squares $[i,i+1] \times [j,j+1]$ ($(i,j) \in \mathbb{Z}^2$), the torus $\manifold$ will also be cellularized by $d\times d$ squares.
\end{ex}

\begin{ex}[Hyperbolic surface codes]
For a hyperbolic surface code, the universal covering space $\unicover = \hpblsp$ is the hyperbolic plane with some regular $\{r,s\}$-tilling, where $\{r,s\}$ is the  \textit{Schl\"{a}fli symbol} indicating that each face is an $r$-gon and $s$ faces meet at each vertex. 
The group of homeomorphisms from $\hpblsp$ to itself that preserve the tilling structure is called the \textit{Coxeter group} and denoted by $G_{r,s}$.
The group $G$ is chosen to be a normal subgroup of $G_{r,s}$ such that $\group$ has no fixed points on $\unicover$ and $|G_{r,s} / G| < \infty$ so that $\unicover / G$ has finitely many faces. See Refs. \cite{breuckmann2016constructions,breuckmann2017homological} for details of the constructions. The parameters $[[n,k,d]]$ satisfy that $k=O(n)$ and $d = O(\log n)$. 
\end{ex}

\subsubsection{Covering spaces}

In practice, we are interested in choices of $G$ satisfying the following requirement: for each 
$u \in \unicover$, $u$ has an open neighborhood $N_u$ such that all the images $g(N_u)$ for varying $g \in G$ are disjoint and homeomorphic to $$N_v := p_G(N_u) \subseteq \unicover / G,$$ which is an open neighborhood of $v := Gu \in \unicover / G$. 
Moreover, the map $p_G$ maps every $g(N_u)$ homeomorphically to $N_v$. One immediate consequence is that the action of $G$ does not have any fixed points, i.e., there does not exist a non-trivial $g \in G$ such that $g(u) = u$ for some $u \in \unicover$.

\begin{ex}
For the $[[2d^2,2,d]]$ toric code, the neighborhood $N_u$ for any $u = (x,y) \in \unicover = \euclidsp$ can simply be chosen as an open disk $D_{1/2}(x,y)$ centered at $(x,y)$ of radius $1/2$. 
For any choice of $t_{r,s}$, the distance between $(x,y)$ and $t_{r,s}(x,y)=(x+dr,y+ds)$ is at least $d \ge 1$ so that $$D_{1/2}(x,y) \cap D_{1/2}(x+dr,y+ds) = \emptyset.$$
\end{ex}

Another way to interpret the above technical requirement
is, for each $v \in \unicover / G$, $v$ has an open neighborhood $N_v$ such that the preimage of $N_v$, $p_G^{-1}(N_v)$, are disjoint copies of $N_v$.
In topology, $\unicover$ is called a \textit{covering space} of $\unicover / G$ with the \textit{covering map} $p_G: \unicover \rightarrow \unicover / G$.
The important property of covering spaces we care about is the so-called \textit{lifting property}: 
for any point $u \in \unicover$,
a loop $\ell$ on $\unicover / G$ starting at $Gu = p_G(u)$
can be uniquely lifted as a path 
$\tilde{\ell}$ on $\unicover$ starting at $u$ such that $p_G$ maps $\tilde{\ell}$ down to $\ell$\footnote{
Formally, the paths $\tilde{\ell}$ and $\ell$ are regarded as maps from $[0,1]$ to $\unicover$ and $\unicover / G$ respectively. By mapping $\tilde{\ell}$ down to $\ell$, we mean $p_G \circ \tilde{\ell} = \ell$.}.
Since $\ell$ is a loop, the endpoint of $\tilde{\ell}$, denoted by $u'$, is also mapped to $Gu$ by $p_G$, i.e., $u' = g(u)$ for some $g \in G$.

\begin{ex}
For the $[[2d^2,2,d]]$ toric code, suppose $u = (0,0) \in \unicover = \euclidsp$. The horizontal loop $\ell$ on the torus $\unicover / G$ starting at $Gu$ is uniquely lifted as a path $\tilde{\ell}$ starting at $(0,0)$ and ending at $u' = (d,0) = t_{1,0}(0,0)$. Therefore $u' = g(u)$ where $g = t_{1,0}$.
\end{ex}

If $\unicover$ is simply-connected by default, an additional property is guaranteed: the lift $\tilde{\ell}$ of $\ell$ is a loop if and only if $\ell$ is contractible on $\unicover / G$. 
One can classify all loops $\ell$ starting at $v \in \unicover / G$ by the endpoints $g(u)$ of their lifts $\tilde{\ell}$. In fact, $G$ is known as the \textit{fundamental group} of $\unicover / G$. See Ref. \cite{Hatcher} for more details. 

Given a subgroup $H$ of $G$, one can define another orbit space $\unicover / H$ with orbits $Hu$ ($u \in \unicover$).
$\unicover / H$ is also a covering space of $\unicover / G$ with a covering map $p_G^H: \unicover / H \rightarrow \unicover / G$ defined by $p_G^H(Hu) = Gu$ for all $u \in \unicover$. 
This is to say that for any $u \in \unicover$, a loop $\ell$ starting at $Gu \in \unicover / G$ can be uniquely lifted to a path $\ell_H$ on $\unicover / H$ starting at $Hu \in \unicover / H$. Note that if $\ell$ is lifted to $\unicover$ as a path $\tilde{\ell}$ starting at $u$ and ending at $g(u)$ for some $g \in G$, the path $\ell_H$ must end at $p_H(g(u)) = H(g(u))$. This is because $\unicover$ is also a covering space of $\unicover / H$, and $\tilde{\ell}$ is a lift of $\ell_H$.
In particular, $\ell_H$ will be a loop if and only if $g \in H$. 
\begin{ex}
For the $[[2d^2,2,d]]$ toric code, suppose $H = \langle t_{1,0} \rangle \subseteq G$ is the group of horizontal translations. The space $\unicover / H$ is a vertical cylinder defined by the product of a horizontal circle of length $d$ and the real line $\mathbb{R}$. A horizontal loop on the torus $\unicover / G$ is lifted as a loop on the cylinder $\unicover / H$. A vertical loop on the torus is lifted as a vertical segment of length $d$.
\end{ex}

In fact, there is a one-to-one correspondence between subgroups $H$ of $G$ and covering spaces $\unicover / H$ of $\unicover / G$. The space $\unicover$, which is the covering space of all $\unicover / H$, is called the \textit{universal cover} of $\unicover / G$.

\subsubsection{Constructions of homomorphic gadgets}
\label{sec:construction}
\begin{figure*}
    \centering
    \includegraphics[width=\linewidth]{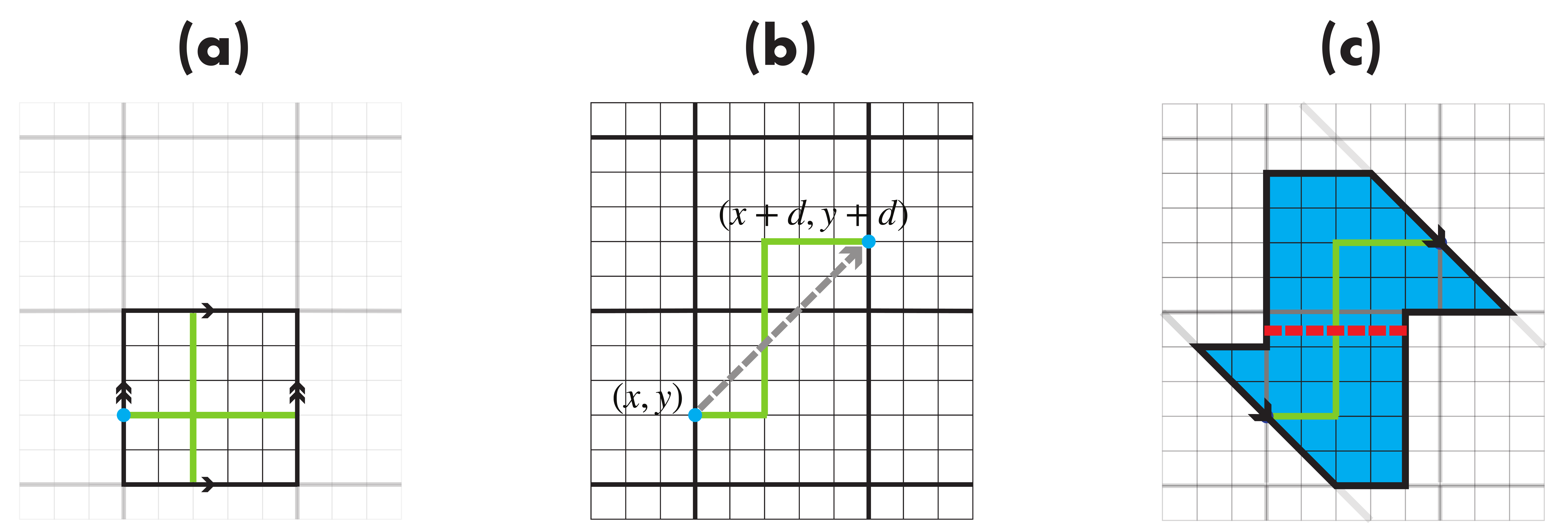}
    \caption{The procedure of constructing the homomorphic gadget for measuring $\overline{Z_1Z_2}$ of the toric code of $d=5$. (a) The green path is the loop $\ell$ corresponding to the logical $\overline{Z_1Z_2}$ operator. A basepoint of $\ell$ is chosen and marked by the blue dot.
    (b) The lift of $\ell$ from the torus $\unicover / G = T^2$ to its universal cover $\unicover = \euclidsp$  is a path connecting two points $(x,y)$ and $(x+d,y+d)$. $(x,y)$ and $(x+d,y+d)$ are related by the translation $t_{1,1} \in G$. 
    (c) The covering space $\unicover / H$ is a cylinder, where $H = \langle t_{1,1} \rangle$. The green path is the lift $\ell'$ of $\ell$ from $\unicover / G$ to $\unicover / H$, which corresponds to the unique  $Z$-type logical operator of $\unicover / H$.
    The ancilla code $\nbhood$ is a chosen to be a neighborhood of $\ell'$ colored by blue.  The logical $X$ operator is formed by edges on the cut colored in red.
    }
    \label{fig:toric_code_Z1Z2}
\end{figure*}

After the above discussion, we are ready to see that
to measure a $Z$-type logical operator of $\data = \unicover / G$ represented by a non-contractible loop $\ell$,
instead of searching for a subcomplex $\ancilla \supset \ell$ of $\data$ directly, we can lift the loop $\ell$ to a loop $\ell'$ on some covering space $p: \widetilde{\data} \rightarrow \data$ that unfolds the other non-contractible loops of $\data$ (except for self-compositions of $\ell$), then find a subcomplex $\ancilla \subset \widetilde{\data}$ 
such that $\ancilla \supset \ell'$
and $d_\ancilla = d_\data$.

As previously mentioned, $\widetilde{\data} = \unicover / H$ for some subgroup $H \subseteq G$ and $p = p_G^H$ is the covering map from $\unicover / H$ to $\unicover / G$.
To determine the subgroup $H \subset G$, suppose the lift of $\ell$ from $\unicover / G$ to the universal cover $\unicover$ is a path from $u \in \unicover$ to $g(u)$ for some $g \in G$.   
f we pick $H = \langle g \rangle$, from the previous discussion, all other loops except for $\ell$ (and its self compositions) will be unfolded when being lifted from $\data = \unicover / G$ to $\widetilde{\data} =  \unicover / H$.
The lift $\ell'$ from $\data$ to $\widetilde{\data}$ will therefore represent the unique logical $Z$ operator of the surface code $\widetilde{\data}$.
However, $\widetilde{\data}$ is infinitely large.
To have a finite-size ancilla code $\ancilla$, we can define $\ancilla \subseteq \widetilde{\data}$ to be some surface code defined on a finite neighborhood of $\ell'$ 
such that $d_\ancilla = d_\data$.
The gate matrix of the homomorphic gadget is induced by the map $\gamma := p \circ \tilde{\gamma}$, where $\tilde{\gamma}: \nbhood \hookrightarrow \widetilde{\data}$ is the inclusion map.

\begin{figure}
    \centering
    \includegraphics[width=\linewidth]{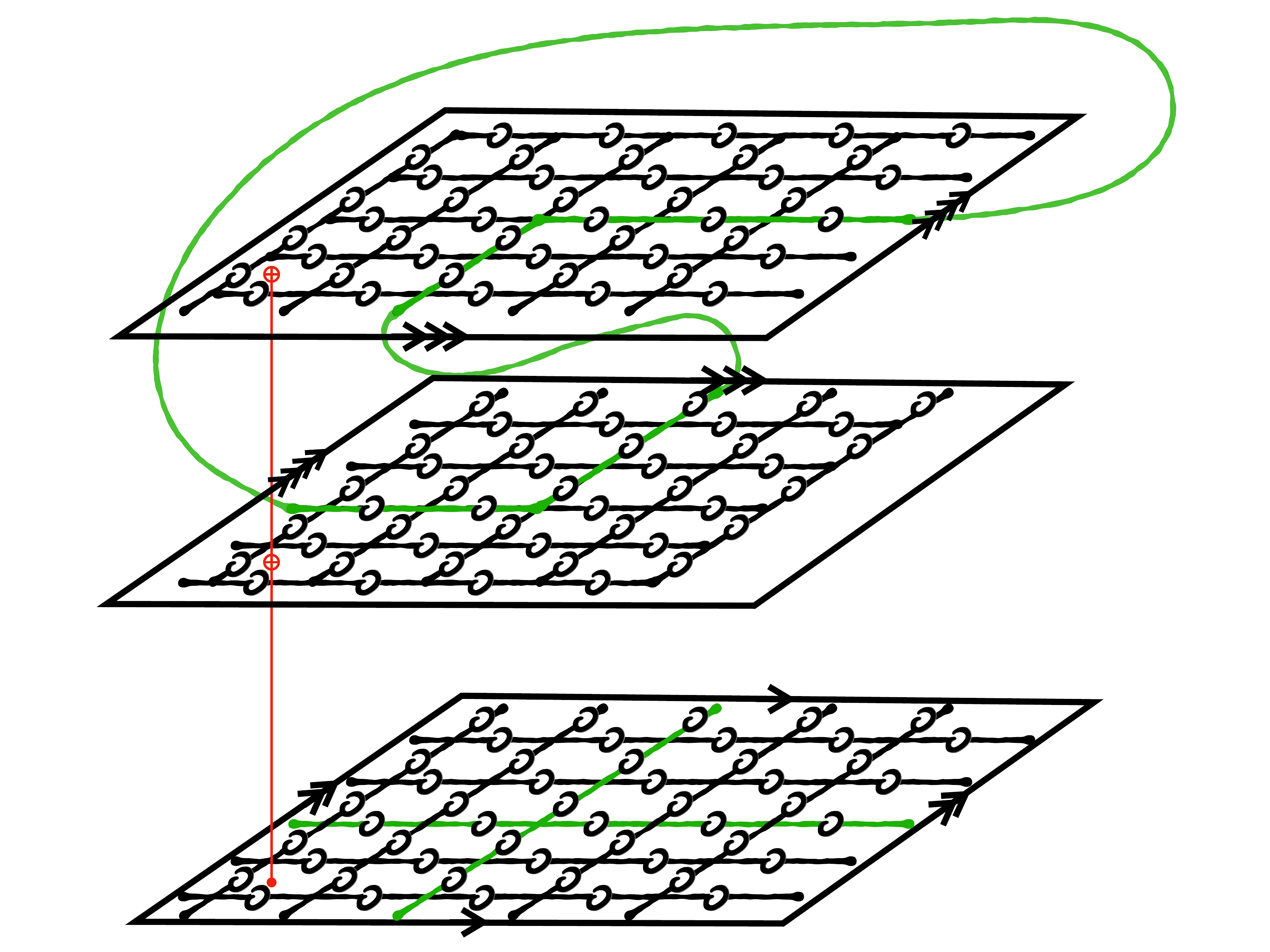}
    \caption{
    Example of the implementation of a homomorphic measurement of logical~$\overline{Z_1Z_2}$ in the toric code. The data code is defined on the torus on the bottom layer. The ancilla code (defined in \fig{toric_code_Z1Z2}(c)) is a ribbon, which when mapped onto the original torus corresponds to two layers with a single logical operator that is given by the green loop. The logical operator is measured by performing a CNOT gate between the data qubits (as control) and each corresponding ancilla qubit (as target).
    }
    \label{fig:CoveringEx}
\end{figure}

\begin{ex}
Suppose we want to measure the logical $\overline{Z_1Z_2}$ of the $[[2d^2,2,d]]$ toric code. As shown in \fig{toric_code_Z1Z2}(a)(b), the loop $\ell$ corresponding to logical $\overline{Z_1Z_2}$ is lifted as a path on $\unicover = \euclidsp$ from $(0,0)$ to $(d,d) = t_{1,1}(0,0)$. Therefore we choose $H = \langle t_{1,1}\rangle$ and $\unicover / H$ is a cylinder. As shown in \fig{toric_code_Z1Z2}(c), $\ell$ is lifted as a loop $\ell'$ along the cylinder, and one can always find a large enough neighborhood $\nbhood$ of $\ell'$ that encodes exactly one logical qubit and has distance $d$. An explicit example is given for the distance-5 toric code in Fig.~\ref{fig:CoveringEx}.

\end{ex}

If we wish to measure $k'$ independent $Z$-type logical operators in parallel, 
we first specify their representatives $\ell_1$, \ldots, $\ell_{k'}$ starting at the same base point. Each loop $\ell_i$ is lifted from $\unicover / G$ to the universal cover $\unicover$ as a path from $u \in \unicover$ to $g_i(u)$ for some $g \in G$. We then choose $H = \langle g_1, \ldots, g_{k'}\rangle$ and $\widetilde{\data} = \unicover / H$.
Let $\ell_i'$ be the lift of $\ell_i$ from $\data = \unicover / G$ to $\widetilde{\data}$. We can define the ancilla code $\ancilla$ to be a subcomplex
of $\widetilde{\data}$ such that $\ell_i' \subset \ancilla$ for all $i$ and $d_\ancilla = d_\data$. The ancilla code $\ancilla$ encodes $k'$ logical qubits and is prepared in $\overline{\ket{0^{\otimes k'}}}$.

\subsubsection{Fault tolerance}

Unlike the homomorphic gadgets defined from subcomplexes of $\data = (\vertexset, \edgeset, \faceset)$, the homomorphic gadgets defined from covering spaces fail to have a transversal data-ancilla interaction, as an edge (data qubit) $e \in \edgeset$  could be covered by multiple edges (ancilla qubits) of $\ancilla = (\vertexset', \edgeset', \faceset')$. A potential issue is that correlated $X$-errors on the ancilla code block $\ancilla$ could be induced by $X$-errors on the data block $\data$.
As an example, an $X$-error occurs on the qubit $e$ before the interaction will propagate to all edges in the subset $\gamma_1^\transpose(e) \subset \edgeset'$.
In \app{fault_tolerance}, we prove that for any gate schedule of the data-ancilla interaction, 
the effective $X$-distance of $\ancilla$ is guaranteed to be $\min\{d_\data, d_\ancilla\}$.
As a result, repetitive measurements are not required for boosting the logical measurement accuracy.

\subsection{Joint measurement}\label{sec:joint_meas}

Now we discuss the case that the measured $Z$-type logical operator $\overline{P}$ corresponds to two disjoint loops $\ell_1$ and $\ell_2$. If the manifold $\manifold$ is path-connected,  
one can always find a path $p$ on $(\vertexset, \edgeset)$ that connects a vertex $v_1$ on $\ell_1$ to another vertex $v_2$ on $\ell_2$, 
then construct a loop $\ell = \ell_1 \cdot p \cdot \ell_2 \cdot \bar{p}$ 
that starts from $v_1$, traverses the paths $\ell_1$, $p$, $\ell_2$ and $\bar{p}$ (the reverse path of $p$), and finally ends at $v_1$. It is clear that $\ell$ also corresponds to the logical opeartor $\overline{P}$ and we can apply the construction discussed in \sec{homgadget_coveringspace}.

\begin{figure}
    \centering
    \includegraphics[width=0.85\linewidth]{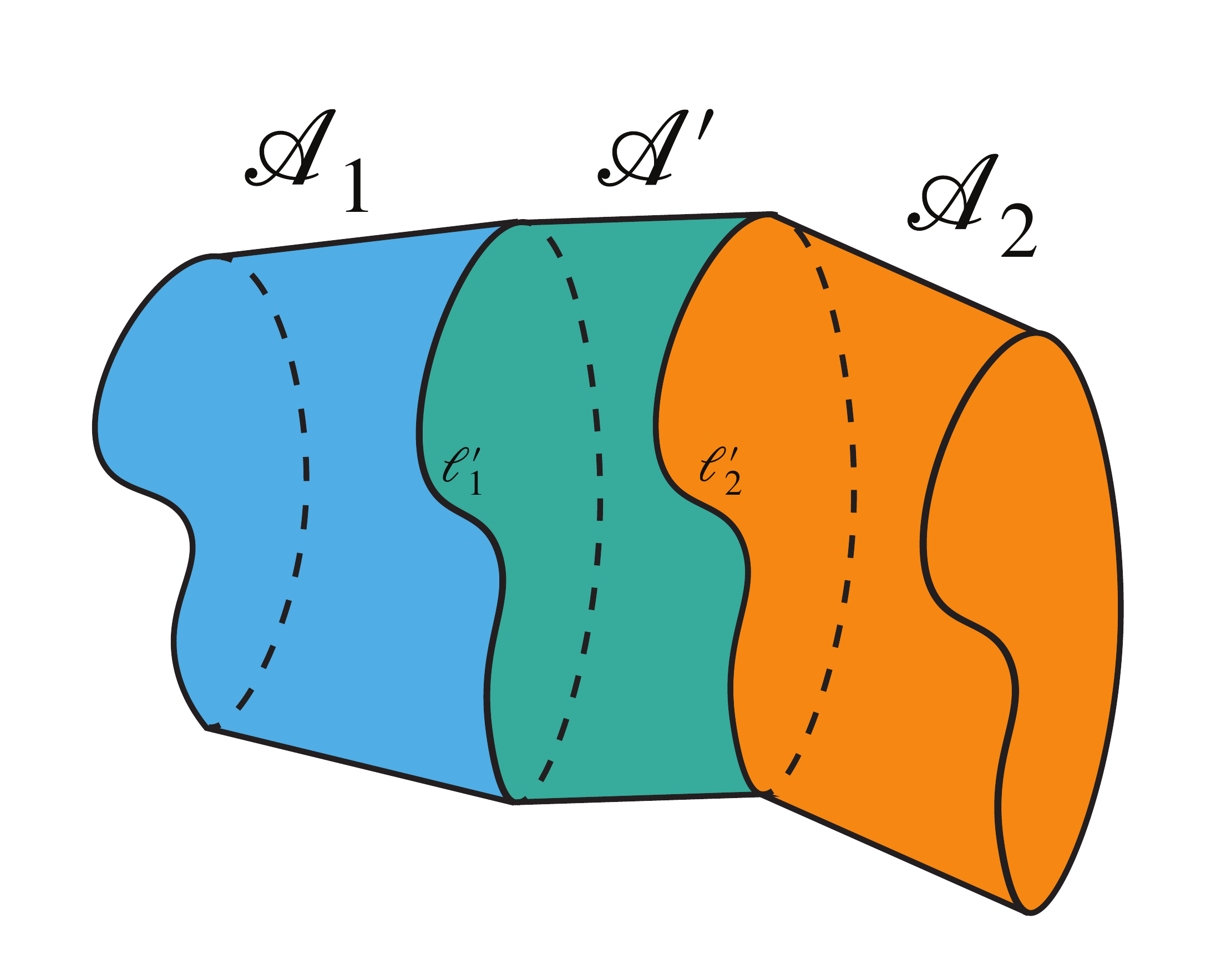}
    \caption{
    Ancilla preparation for jointly measuring the product of two $Z$-type logical operators $\overline{P}_1$ and $\overline{P}_2$.
    $\ancilla_1$ (in blue) and $\ancilla_2$ (in orange) are homomorphic ancilla blocks for measuring $\overline{P}_1$ and $\overline{P}_2$ respectively, both of which are surface codes encoding one logical qubit. Each $\ancilla_i$ has some smooth boundary $\ell_i'$ representing its logical $Z$ operator.
    To measure $\overline{P_1P_2}$ without measuring each individual $\overline{P}_i$, the joint code $\ancilla_1 \cup \ancilla_2$ should be prepared in logical Bell state. One can first prepare each $\ancilla_i$ in logical $\overline{\ket{+}}$, then use some surface code $\ancilla'$ (in green) to merge $\ell_1'$ and $\ell_2'$ and split afterwards, which will measure $\overline{Z_1Z_2}$ of $\ancilla_1 \cup \ancilla_2$. If the measurement outcome is $-1$, one fixes the outcome by applying $\overline{X}_1$ on $\ancilla_1$.
    }
    \label{fig:joint}
\end{figure}

If $\manifold$ is not path-connected (which happens, for instance, in the case of two separate toric codes), the above trick does not work. Suppose $\nbhood_1$ and $\nbhood_2$ are the two ancilla surface codes for measuring the $Z$-type logical operators corresponding to $\ell_1$ and $\ell_2$, respectively. Our goal is to jointly prepare the ancillas $\nbhood_1$ and $\nbhood_2$ in the logical Bell state stabilized by the logical $\overline{Z_{\nbhood_1}Z_{\nbhood_2}}$ operator.
Each $\nbhood_i$ ($i = 1,2$) is not closed and has smooth boundaries. Thus the logical $Z$ operator of $\nbhood_i$ can be represented by some loop $\ell_i' \subseteq \partial \nbhood_i$ along the boundary of $\nbhood_i$. To prepare
the ancilla logical Bell state, one can first prepare $\overline{\ket{+}}_{\nbhood_1}$ and $\overline{\ket{+}}_{\nbhood_2}$ seperately, then apply lattice surgery~\cite{Horsman_2012} to measure $\overline{Z_{\nbhood_1}Z_{\nbhood_2}}$, by joining the two boundary loops $\ell_1'$ and $\ell_2'$ by some surface code $\nbhood'$ satisfying $\partial \nbhood' = \ell_1' \cup \ell_2'$, which is illustrated in \fig{joint}. 

\subsection{Comparison with lattice surgery approaches}

Previously, logical Pauli measurements of hyperbolic surface codes are accomplished by lattice surgery approaches~\cite{Breuckmann2017hyperbolic,cohen2022lowoverhead}.
For example, suppose we want to measure a $Z$-type logical operator $\overline{P}$. We can prepare an ancilla planar surface code in $\overline{\ket{0}}$ state, then jointly measure $\overline{P} \overline{Z_{{\rm anc}}}$ by merging the data block and the ancilla block then splitting, where $\overline{Z_{{\rm anc}}}$ is the logical $Z$ operator of the ancilla code. If $\overline{P}$ has weight $w$, the size of the ancilla surface code will be $O(dw)$. The lattice surgery procedure takes $d$ repetitive syndrome measurement rounds to guarantee fault tolerance.
In addition, when the data and ancilla blocks are merged, 
each qubit supporting $\overline{P}$ will be adjacent to $3$ faces instead of $2$. An $X$ error on these qubits will create $3$ excitations and thus the decoding graph becomes a hypergraph. As a result, convential surface code decoding algorithms such as MWPM~\cite{Dennis:2002,edmonds2009paths} can not be applied immediately.

As a comparison, the ancilla states for homomorphic measurements can be prepared \textit{offline} by $d$ rounds of Shor-style syndrome measurements. More importantly, the structure of the data code is unchanged so that MWPM can be still applied.
However, the size of the ancilla block could be a disadvantage:
In \app{upperbound}, we argue that the ancilla block has size $O(nw/d)$.
For hyperbolic surface codes, 
as $d = O(\log n)$ and $w = O(d)$, 
the ancilla size will be $O(n)$, 
while the ancilla of the lattice surgery approach will have size $O(\log^2 n)$.
If ancilla size is a concern, one can choose an ancilla of $X$-distance $d/c$ for some number $c$ independent of $d$, and repeat the logical measurement circuit for $O(c)$ rounds.

\subsection{Surface codes with rough boundaries}\label{sec:roughboundary}

\begin{figure}
    \centering
    \includegraphics[width=0.82\linewidth]{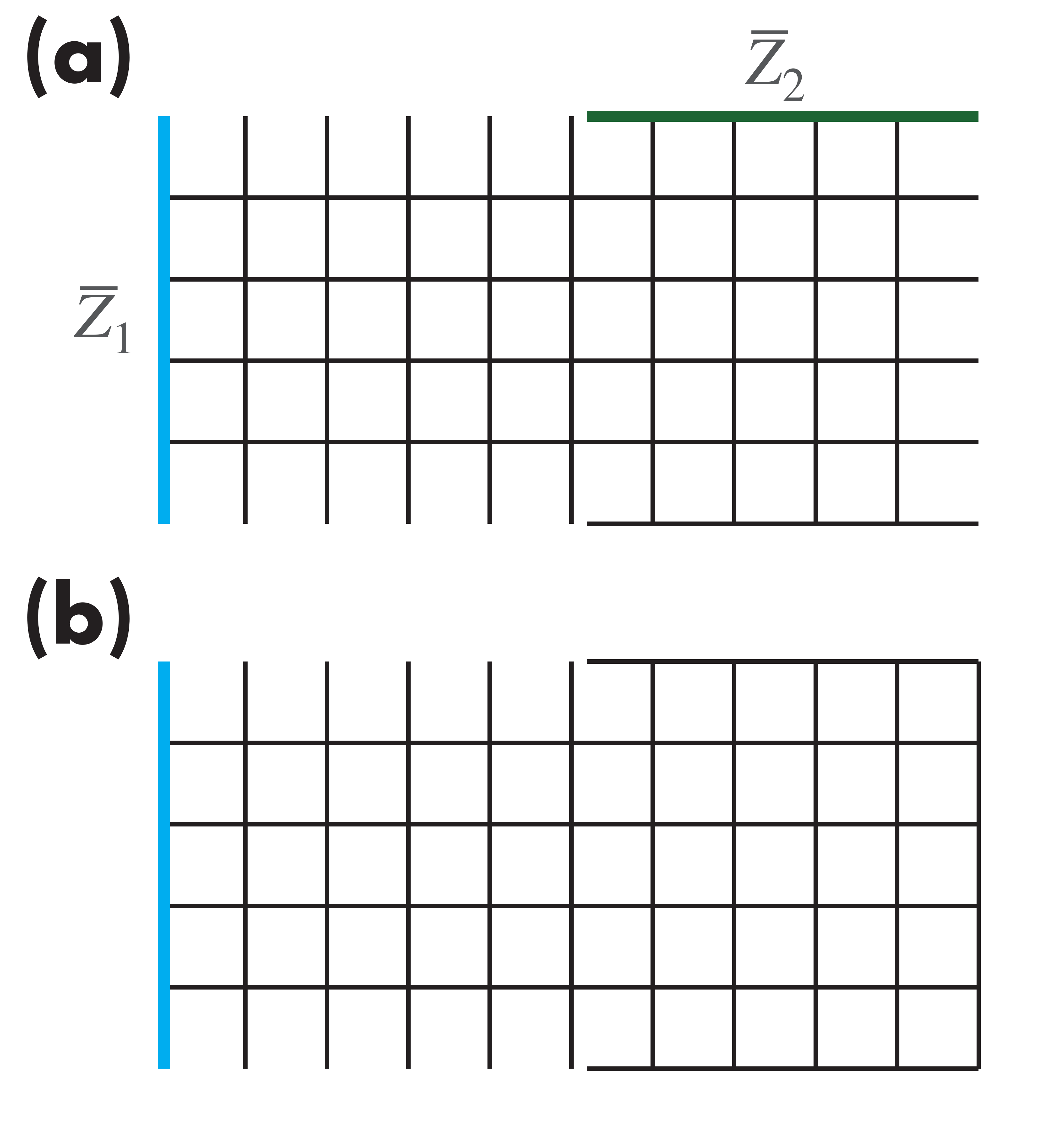}
    \caption{(a) A planar surface code encoding two logical qubits. The logical $\overline{Z}_1$ and $\overline{Z}_2$ operators are represented by blue and green strings. (b) Ancilla patch for measuring $\overline{Z}_1$. The rough boundary on the right is closed by adding qubits and $X$-checks. The ancilla patch encodes $1$ logical qubit and is prepared in $\overline{\ket{0}}$. A CNOT gate is applied between every data qubit (as control) and its corresponding ancilla qubit (as target).
    }
    \label{fig:boundary}
\end{figure}

For surface codes with rough boundaries, instead of being a closed loop, a $Z$-type logical operator can be an open path connecting two different rough boundaries. 
If the surface is planar and without holes, as shown in \fig{boundary}, 
the ancilla code is almost identical to the data code, while the rough boundaries that do not support the measured $Z$-type logical operator are all being closed by adding extra edges and vertices and therefore become smooth boundaries.
If non-trivial loop operators exist on the surface, we can in principle combine the trick above and the procedure discussed in \sec{homgadget_coveringspace} to unfold the unwanted logical loops. As pointed out in Ref.~\cite{breuckmann2016constructions}, it is unlikely that surface code with boundaries will have good parameters $k$ and $d$. Therefore we do not provide a rigorous treatment for surface codes with both boundaries and holes.

\section{Conclusions and outlook}\label{sec:conclusion}

In this paper, we propose a framework called \textit{homomorphic measurement} to generalize the traditional Shor and Steane measurements and circumvent their individual difficulties on quantum LDPC codes. For surface codes such as the toric code~\cite{kitaev2003fault} and hyperbolic surface codes~\cite{breuckmann2016constructions}, we provide a procedure to construct homomorphic gadgets for single-shot non-destructive logical CSS  measurements. The ancilla code blocks can be prepared off-line without any state distillation. In particular, our scheme preserves the 2D structure of the surface code being measured and therefore does not complicate the decoding process.
Our scheme can serve as a complement of the existing logical operation approaches, including lattice surgery~\cite{Breuckmann2017hyperbolic,cohen2022lowoverhead} and Dehn twists~\cite{Breuckmann2017hyperbolic,zhu2020universal,lavasani2019universal}.

We note that similar to gate teleportations~\cite{gottesman1999demonstrating}, one can initialize the ancilla code in some magic states to implement non-Clifford gates on particular logical qubits of the data. While magic state distillation protocols~\cite{bravyi2005universal,bravyi2012magic} are still required for preparing these states,
as the ancilla code only encodes one logical qubit, the full distillation procedure will be simplified.

Our construction of homomorphic gadgets for surface codes are inspired from geometrical intuitions. However,
many recent quantum LDPC codes of good performance are constructed from products of classical codes~\cite{tillich2014quantum,hastings2021fiber,panteleev2002asymptotically,breuckmann2021balanced,panteleev2021quantum,leverrier2022quantum}, which are purely algebraic. 
The standard toric code can also be obtained by product of two classical repetition codes~\cite{tillich2014quantum}. However, as we have seen in \sec{homgadget_coveringspace}, 
the homomorphic gadget for measuring $\overline{Z}_1\overline{Z}_2$ does not have a product structure. Thus, it is unlikely that for codes with product constructions, useful homomorphic gadgets have product structure as well. 
We hope that these gadgets can still be constructed and inherit the decoding properties of the original code such as single-shot error correction~\cite{Campbell2019, Fawzi2018constant}.

Finally, it is known that Steane measurement can be generalized for arbitrary stabilizer codes \cite{yoder2018practical,Zheng_2018,Zheng_2020}, while Shor and Knill \cite{knill2005scalable} measurements work for non-CSS situations directly.
Generalizations of our framework to non-CSS Pauli measurements is a natural direction, which is left for future work.

\section*{Acknowledgements}

The authors would like to thank Benjamin Brown, Kenneth Brown, Andrew Cross and Guanyu Zhu for valuable discussions and feedback. 
S.H. is supported by the NSF EPiQC Expeditions in Computing (1832377), the Office of the Director of National Intelligence - Intelligence Advanced Research Projects Activity through an Army Research Office contract (W911NF-16-1-0082), and the Army Research Office (W911NF-21-1-0005).
S.H. also thanks the IBM graduate internship program for support. T.~J. and T.~Y. are supported by the U.S. Department of Energy, Office of Science, National Quantum Information Science Research Centers, Co-design Center for Quantum Advantage (C2QA) contract~(DE-SC0012704).
\bibliography{refs}

\appendix

\begin{figure*}\label{fig:upperbound}
\includegraphics[width=\textwidth]{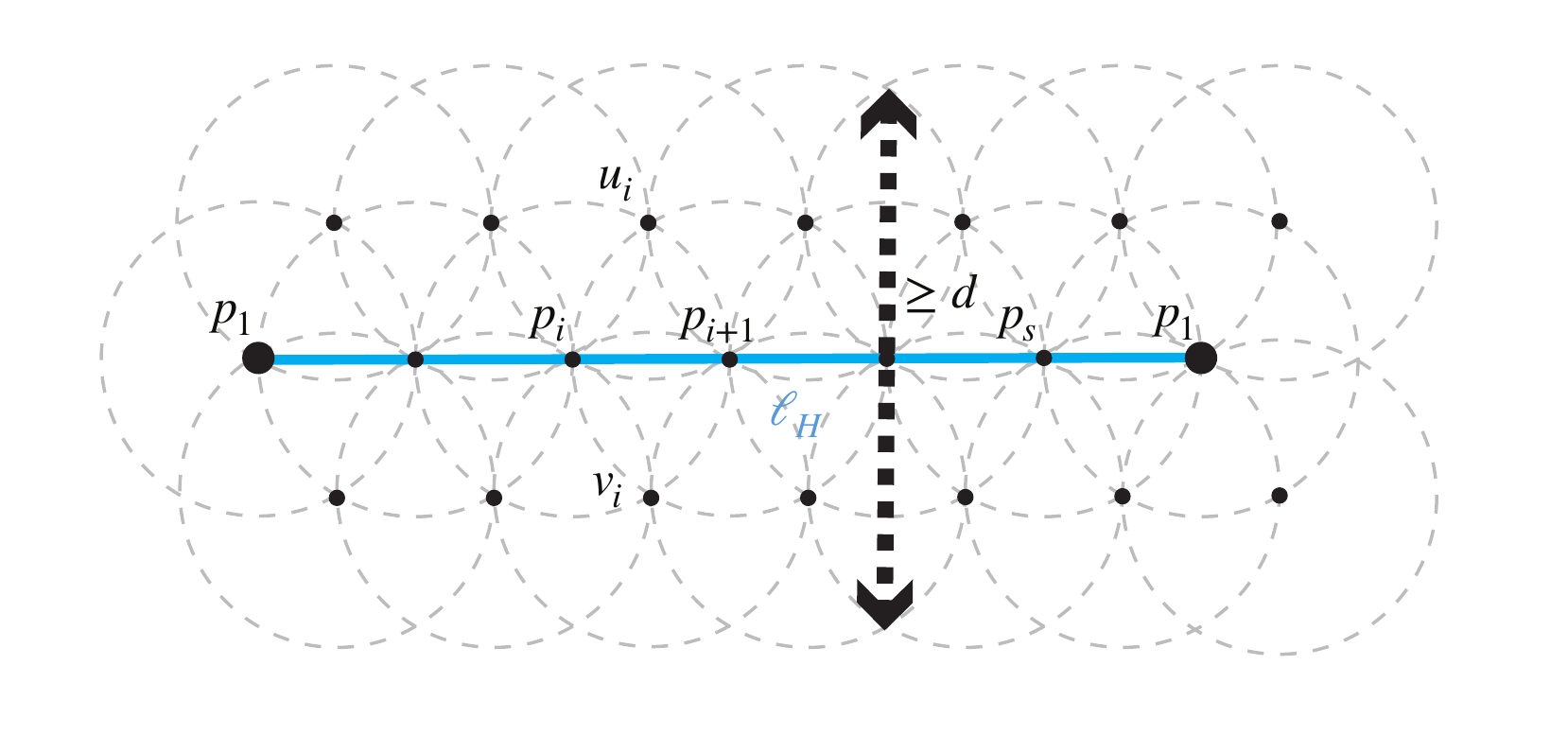}
\caption{A construction of neighborhood of the loop $\ell_H$ that has size $O(nw/d)$. Each disk has radius $d/2$.}
\end{figure*}

\section{Fault tolerance of homomorphic gadgets}\label{app:fault_tolerance}

Suppose a homomorphic gadget for a surface code $\data = (\vertexset, \edgeset, \faceset)$ of distance $d_\data$ is constructed following the procedure described in \sec{construction}, with an ancilla surface code $\ancilla = (\vertexset', \edgeset', \faceset')$ of distance $d_\ancilla$ and a gate matrix $\gamma_1$.
Let $E_\ancilla$ be any representative of the logical $X$ operator of $\ancilla$.
$E_\ancilla$ can be considered as a subset of $\edgeset'$.
Define $E_\data \subseteq \edgeset$ be the subset of edges such that $e \in E_\data$ if and only if $\gamma_1^\transpose(e) \cap E_\ancilla \ne \emptyset$. 
In the worst case, $X$-errors on each $e \in E_\data$ spread to all ancilla qubits in the set $E_\ancilla$, and the logical error $E_\ancilla$ can be effected with $|E_\data|$ errors instead of $|E_\ancilla|$. 

To guarantee fault tolerance, we show that $|E_\data|
\ge \min\{d_\data, d_\ancilla\}$.
Note that it suffices to consider the case that $E_\ancilla \subseteq \edgeset'$ is \textit{connected}, which is defined in the following sense:
\begin{dfn}[Connected edge subset]
Suppose we have a subset of edges $E \subseteq \edgeset$ on a cell complex $(\vertexset, \edgeset, \faceset)$. We say that $E$ is \textbf{connected}, if for any two edges $e,e'\in E$, one can find a sequence of edges $e_1, \cdots, e_k \in E$ such that
(i) $e_1 = e$ and $e_k = e'$, (ii) for $i = 1, \ldots, k-1$, $e_i$ shares a face $f \in \faceset$ with $e_{i+1}$.
\end{dfn}

For any non-connected $E_\ancilla$, one can always replace $E_\ancilla$ by
one of its connected component $E_\nbhood'$ which also represents the logical $X$ operator. The corresponding edge set $E_\data' \subseteq E_\data$ has fewer edges.

\begin{prop}
If $E_\nbhood$ is connected, $|E_\data| \ge \min\{d_\data, d_\ancilla\}$.
\end{prop}
\begin{proof}
If $|\gamma_1^\transpose(e) \cap E_\ancilla| = 1$ for all $e \in E_\data$, as $|E_\ancilla| \ge d_\ancilla$, we immediately have $|E_\data| \ge d_\ancilla$.

Suppose there exist one $e \in E_\data$ such that $|\gamma_1^\transpose(e) \cap E_\nbhood| > 1$. 
Any edge sequence connecting two different edges in the subset $\gamma_1^\transpose(e) \cap E_\nbhood$ must have at least $d_\data+1$ edges. 
Otherwise, the edge sequence will be mapped down as a cut (non-trivial $X$-type logical operator) of $\data$ with weight less than $d_\data$, which contradicts to the definition of code distance.
Moreover, the edges in the edge sequence are mapped to at least $d_\data$ different edges in $\edgeset$ by $\gamma_1$. Therefore one must have $|E_\data| \ge d_\data$.
\end{proof}

\section{Upper bound of the ancilla size}\label{app:upperbound}

In this section, we argue that the ancilla block constructed in \sec{homgadget_coveringspace} has size $O(nw/d)$, where $w$ is the weight of the measured logical operator.

Let $\ell_H$ be the lift of the logical loop $\ell$, which is included by the ancilla block $\ancilla$. We want to construct a neighborhood of $\ell_H$ that supports $O(nw/d)$ qubits and has $X$-distance at least $d$.
Pick a set of points $p_1$, \ldots, $p_s$ on $\ell_H$, where $p_i$ and $p_{i+1}$ (also $p_s$ and $p_1$) are separated by a distance less than $d/2$. We can have $s = O(w/d)$. Let $D_{d/2}(p_i)$ be the disk centered at $p_i$ of radius $d/2$. Each disk $D_{d/2}(p_i)$ covers no more than $n$ qubits as there are no two qubits inside $D_{d/2}(p_i)$ separated by a distance $d$. The union $\bigcup_i D_{d/2}(p_i)$ covers $O(nw/d)$ qubits. However, if we simply pick the largest surface code inside $\bigcup_i D_{d/2}(p_i)$, its distance is not reaching $d$: the distance between the intersection points $u_i$ and $v_i$ between two circles $\partial D_{d/2}(p_i)$ and $\partial D_{d/2}(p_{i+1})$ is less than $d$. As shown in \fig{upperbound}, the code distance can achieve $d$ if the neighborhood is chosen to be
\begin{eqnarray}
\bigcup_{i=1}^s \left(D_{d/2}(p_i) \cup D_{d/2}(u_i) \cup D_{d/2}(v_i) \right).
\end{eqnarray}
\end{document}